\pdfoutput=1

\RequirePackage{snapshot}

\RequirePackage{sty/bootstrap}
\AppendInputPath{{sty/}{sty/img/}{fig/}{build/}}

\RequirePackage[twocolumn]{IEEEtranSetup}
\documentclass[conference,final]{IEEEtran}

\ProvideIf[true]{draft}
\IfClassLoadedWithOptionsTF{IEEEtran}{final}{\draftfalse}{}

\ProvideIf[false]{ieee}  
\ProvideIf[false]{ieeepreview}
\ProvideIf[false]{blind} 
\ProvideIf[true]{arxiv} 

\ifieee 
  \usepackage[hyperref=false,bibtex]{boilerplate}
\else
  \ifarxiv
    \usepackage[biblatex]{boilerplate}
    
  \else
    \usepackage[bibtex]{boilerplate}
  \fi
\fi

\ProvideDocumentCommand{\orcidlink}{m}{}

\usepackage{csquotes}
\usepackage[capitalize]{cleveref}
\crefname{equation}{}{}
\Crefname{equation}{Equation}{Equations}

\usepackage{nameref}

\NewDocumentCommand \arxivOrcid {sm} {%
  \IfBooleanTF{#1}{,\hspace{-3pt}}{}%
  \ifarxiv\orcidlink{#2}\fi%
}

\title{
  Capacity Bounds for \\
  Identification With Effective Secrecy
}

\date{\today}
\ifblind\else
\author{
  \IEEEauthorblockN{
    Johannes~Rosenberger\arxivOrcid*{0000-0003-2267-3794}\IEEEauthorrefmark{1}
    Abdalla~Ibrahim\arxivOrcid*{0009-0008-7646-6276}\IEEEauthorrefmark{1}
    Boulat~A.~Bash\arxivOrcid*{0000-0002-1205-3906}\IEEEauthorrefmark{2}
    Christian~Deppe\arxivOrcid*{0000-0002-2265-4887}\IEEEauthorrefmark{1}
    Roberto~Ferrara\arxivOrcid*{0000-0002-1991-3286}\IEEEauthorrefmark{1}
    Uzi~Pereg\arxivOrcid{0000-0002-3259-6094}\IEEEauthorrefmark{3}
  }
  \IEEEauthorblockA{
    \IEEEauthorrefmark{1}
    TUM School of Computation, Information and Technology,
    Technical University of Munich,
    \\
    Email: \{\tumail{johannes.rosenberger}
      , \tumail{abdalla.m.ibrahim}
      , \tumail{christian.deppe}
      , \tumail{roberto.ferrara}
    \}@tum.de
  }
  \IEEEauthorblockA{
    \IEEEauthorrefmark{2}
    Electrical and Computer Engineering Department, University of Arizona,
    Email: \uref{mailto:boulat@arizona.edu}{boulat@arizona.edu}
  }
  \IEEEauthorblockA{
    \IEEEauthorrefmark{3}
    Faculty of Electrical and Computer Engineering and
    Helen Diller Quantum Center, \\
    Technion---Israel Institute of Technology,
    Email: \uref{mailto:uzipereg@technion.ac.il}{uzipereg@technion.ac.il}
  }
}
\fi

\makeatletter
\hypersetup{
  pdftitle = \string\@title,
  pdfauthor = {},
  pdfkeywords = {},
  pdflang = {English},
}
\makeatother

\newcommand{\code}{\cC}
\NewDocumentCommand{\defname}{om}{\emph{#2}\IfValueT{#1}{~(#1)}}

\newcommand{\capacity}{\textsf{C}}

\newcommand{\capESID}{\capacity_\textsf{ESID}}
\newcommand{\capDI}{\capacity_\textsf{DI}}

\DeclareMathOperator{\BEC}{BEC}
\DeclareMathOperator{\BSC}{BSC}

\usepackage{tikz}
\newif\ifqtikz\qtikzfalse
\usetikzlibrary{math}
\usetikzlibrary{fpu}
\pgfkeys{/pgf/fpu}
\pgfkeys{/pgf/fpu/output format=fixed}
\tikzmath{
  function fold(\p, \q) { return \p*(1-\q) + \q*(1-\p); };
  function Hbin(\p) { return - \p * log2(\p) - (1 - \p)*log2(1-\p); };
  function Ibsc(\p, \q) { return Hbin(fold(\p, \q)) - Hbin(\q)); };
  function Ibec(\p, \q) {
    return Hbin(\p) - \q * Hbin(\p);
  };
  function IgaussC(\s) { return log2(1 + \s); };
  function IgaussR(\s) { return IgaussC(\s)/2; };
  function IzChannel(\p, \d) { return Hbin(\p * \d) - \p * Hbin(\d); };
  function zChannelPopt(\d) {
    real \g;
    let \g = (1-\d)^((1-\d)/\d);
    return \g/(1 + \d*\g);
  };
}
\pgfkeys{/pgf/fpu=false}

\providecommand \ifqtikz {\iftrue}

\usetikzlibrary{fpu} 
\usetikzlibrary{math} 
\usetikzlibrary{datavisualization} 
\usetikzlibrary{datavisualization.formats.functions}

\ifqtikz

\usetikzlibrary{math}
\usetikzlibrary{fpu}
\pgfkeys{/pgf/fpu}
\pgfkeys{/pgf/fpu/output format=fixed}
\tikzmath{
  function fold(\p, \q) { return \p*(1-\q) + \q*(1-\p); };
  function Hbin(\p) { return - \p * log2(\p) - (1 - \p)*log2(1-\p); };
  function Ibsc(\p, \q) { return Hbin(fold(\p, \q)) - Hbin(\q)); };
  function Ibec(\p, \q) {
    return Hbin(\p) - \q * Hbin(\p);
  };
  function IgaussC(\s) { return log2(1 + \s); };
  function IgaussR(\s) { return IgaussC(\s)/2; };
  function IzChannel(\p, \d) { return Hbin(\p * \d) - \p * Hbin(\d); };
  function zChannelPopt(\d) {
    real \g;
    let \g = (1-\d)^((1-\d)/\d);
    return \g/(1 + \d*\g);
  };
}
\pgfkeys{/pgf/fpu=false}

\tikzset{
  font=\sffamily,
  rateBound/identification/.style = {color=identification},
  rateBound/transmission/.style = {dotted,color=transmission}
}

\fi

\newcommand \tikzMutInfsRevDegraded [3][]{

\providecolor{identification}{named}{blue}
\providecolor{transmission}{named}{red}

\tikzmath{
 real \ticklen, \perasure, \pcrossover, \capBSC;
 \ticklen = 0.04; 
 \perasure = #3;
 \pcrossover = #2;
 function IbscBec(\a, \p, \e) { return (1-\e) * ( Hbin(fold(\a,\p)) - Hbin(\p) ); };
}

\begin{tikzpicture}[auto,thick,#1]
  \datavisualization [school book axes={unit=0.25},
    visualize as smooth line/.list={Ixy,Ixz,Iuy,Iuz,Ix1y,It},
    x axis={ticks={step=0.5},label={$p_{X_2}$}},
    y axis={ticks={step=0.5},label={bits/channel use}},
    legend={above right of={x=0.9,y=0.63}},
    style sheet=strong colors,
    style sheet=vary dashing,
    Ixy = {label in legend={text=$I(X;Y)$ [secret ID]}},
    Ixz = {label in legend={text=$I(X;Z)$}},
    Iuy = {label in legend={text=$I(U;Y)$ [ESID]}},
    Iuz = {label in legend={text=$I(U;Z)$}},
    Ix1y = {label in legend={text=$I(X_1;Y)$ [ES transmission]}},
    It = {label in legend={text=$I(U;Y) - I(U; Z)$}},
    data/format=function,
  ]
  data [set=Ixy] {
    var x : interval [0.001:1];
    func y = Ibec(0.5, \perasure) + Ibec(\value x, \perasure);
  }
  data [set=Ixz] {
    var x : interval [0.001:0.999];
    func y = Hbin(\value x);
  }
  data [set=Iuy] {
    var x : interval [0.001:0.999];
    func y = Ibec(0.5, \perasure) + IbscBec(\value x, \pcrossover, \perasure);
  }
  data [set=Iuz] {
    var x : interval [0.001:0.999];
    func y = Ibsc(\value x, \pcrossover);
  }
  data [set=Ix1y] {
    var x : interval [0.001:0.999];
    func y = Ibec(0.5, \perasure);
  }
  data [set=It] {
    var x : interval [0.001:0.999];
    func y = Ibec(0.5, \perasure) + IbscBec(\value x, \pcrossover, \perasure) - Ibsc(\value x, \pcrossover);
  };
\end{tikzpicture}
}

\ifqtikz
\tikzMutInfsRevDegraded{1/8}{0.6866}
\fi

\providecommand \ifqtikz {\iftrue}

\usetikzlibrary{positioning,fit}

\ifqtikz
  \tikzset{
    box/.style={draw, minimum height = .8cm, minimum width = 1cm, inner sep=4pt}, 
    pfeil/.style={->, >=latex}
  }
  \providecommand{\BEC}{\mathrm{BEC}}
  \providecommand{\BSC}{\mathrm{BSC}}
\fi

\NewDocumentCommand \exampleChannelESID { O{} O{1.5} }{

\begin{tikzpicture}[auto,#1]

  \node[box,minimum height=2*#2\baselineskip+2em] (Pu) {$P_{U}$};
  \coordinate[yshift = #2\baselineskip] (Pu1) at (Pu.east);
  \coordinate[yshift = -#2\baselineskip] (Pu2) at (Pu.east);
  \path
    (Pu1)
    ++(4,0)
    node[box](BEC1)  {$\BEC_\epsilon$}

    (Pu2)
    ++(1.6,0)
    node[box] (Pxu2) {$\BSC_q$}
    ++(2.4,0)
    node[box] (BEC2) {$\BEC_\epsilon$}

    (Pu)
   ++(6.5,0)
    node[box,minimum height=2*#2\baselineskip+2em] (Bob) {Bob}
    (Bob.south)
    ++(0,-1)
    node[box] (Willie) {Willie}
  ;

  \coordinate[yshift = #2\baselineskip] (Bob1) at (Bob.west);
  \coordinate[yshift = -#2\baselineskip] (Bob2) at (Bob.west);

  \draw[pfeil] (Pu1) -> (BEC1);
  \draw[pfeil] (Pu2) -> node[above] {$U_2$} (Pxu2);
  \draw[pfeil] (Pxu2.east) -> node[inner sep=0,minimum
          width=4pt,fill,circle,anchor=center] (tap) {} (BEC2);
  \node[above,yshift = 2*#2\baselineskip] (X1) at (tap) {$X_1$};
  \node[above] (X2) at (tap) {$X_2$};
  \draw[pfeil] (tap.center) |- node[above,pos=0.93] {$Z$} (Willie);
  \draw[pfeil] (BEC1.east) -> node[above] {$Y_1$} (Bob1);
  \draw[pfeil] (BEC2.east) -> node[above] {$Y_2$}(Bob2);
  \node[box,fit=(Pu) (Pxu2)] {};

\end{tikzpicture}
}

\ifqtikz
  \exampleChannelESID
\fi

\colorlet{identification}{blue}
\colorlet{transmission}{red}
\tikzset{
	box/.style={draw, minimum height = 1em, minimum width = 1em, inner sep=4pt},
	pfeil/.style={->, >=latex},
  font = {\sffamily\footnotesize}, 
}

\begin{document}

\begin{anfxnote}[author=JR]{}
  \section*{README}
  \label{README}
  \raggedright
  Please do not use \verb|\mathbf|, which does not work for all (math) symbols,
  but instead use the command \verb|\bm|, which is more flexible.
  Mixing both leads to inconsistent output,
  e.g. \verb|\mathbf{X, \alpha}| $\mapsto \mathbf{X, \alpha}$,
  while \verb|\bm{X, \alpha}| $\mapsto \bm{X, \alpha}$.
  The files \verb|sty/shorthands.sty| and \verb|local-shorthands.tex|
  define many shorthands, e.g.
  
  \begin{tabular}{ll}
    \verb|\bX| & $\bX$ \\
    \verb|\bbx| & $\bbx$ \\
    \verb|\cX| & $\cX$ \\
    \verb|\bcX| & $\bcX$ \\
    \verb|\bbN| & $\bbN$ \\
    \verb|\bbR| & $\bbR$ \\
    \verb|\expect| & $\expect$ \\
    \verb|\One| & $\One$ \\
    \verb|\ind{a}| & $\ind{a}$ \\
    \verb|\set{a}| & $\set{a}$ \\
    \verb|\tup{a}| & $\tup{a}$ \\
    \verb|\paren{a}| & $\paren{a}$ \\
    \verb|\intv{a}| & $\intv{a}$ \\
    \verb|\brack{a}| & $\brack{a}$ \\
    \verb|\abs{a}| & $\abs{a}$ \\
    \verb|\card{\cX_1}| & $\card{\cX_1}$ \\
    \verb|\code| & $\code$ \\
    \verb|\capacity| & $\capacity$ \\
    \verb|\capDI| & $\capDI$.
  \end{tabular}
  
  \noindent
  Please define your own shorthands in \verb|local-shorthands.tex|.
\end{anfxnote}

\maketitle

\begin{abstract}
An upper bound to the identification capacity of discrete memoryless wiretap channels is derived under the requirement of semantic effective secrecy, combining semantic secrecy and stealth constraints. A previously established lower bound is improved by applying it to a prefix channel, formed by concatenating an auxiliary channel and the actual channel. The bounds are tight if the legitimate channel is more capable than the eavesdropper's channel. An illustrative example is provided for a wiretap channel that is composed of a point-to-point channel, and a parallel, reversely degraded wiretap channel. A comparison with results for message transmission and for identification with only secrecy constraint is provided.
\end{abstract}

\section{Introduction}
An increasing need for task-oriented and \emph{semantic} communication paradigms
that cater for a variety of tasks with different reliability, robustness,
secrecy and privacy requirements can be observed across modern cyber-physical
systems\cite{Guenduez2023semantic,CabreraEA2021postShannon6G,RezwanCabreraFitzek2022funcomp}. Pioneering work by Shannon \cite{Shannon1948MathematicalTheoryCommunication} emphasized the problem of transmitting messages, where a decoder is expected to decide which message has been sent over a noisy channel among exponentially many possible hypotheses. 

Identification (ID)~\cite{Yao1979SomeComplexityQuestions_conference,
  JaJa1985Identificationiseasier_conference,
  AhlswedeDueck1989Identificationviachannels}
is  a communication task where the receiver, Bob, selects one
of many possible messages and tests whether this particular message
was sent by the transmitter, Alice or not.
It is assumed that Alice does not know Bob's chosen message;
otherwise she could simply answer with “Yes” or “No”, by sending a single bit.
Given the nature of this very specific task,
ID is in stark contrast to the
conventional and general task of uniquely decoding messages,
i.e. estimating which message was sent.
Decoding is general in the sense that Bob can estimate any

function of Alice's message,
while in ID, he can only compute whether
it is the one of his interest or not, hence the function that Bob interested in is a simple indicator function.
However, while the code sizes may only grow exponentially  in the blocklength for the
message-transmission task~\cite{Shannon1948MathematicalTheoryCommunication,
  Yao1979SomeComplexityQuestions_conference,
  JaJa1985Identificationiseasier_conference,
  AhlswedeDueck1989Identificationviachannels},
a \emph{doubly exponential} growth can be achieved in ID,
if randomized encoding is used~\cite{AhlswedeDueck1989Identificationviachannels}.
Randomized encoding allows the number of messages to be restricted not by number of possible distinct codewords, but rather by the number of input distributions to the channel that
are pairwise distinguishable at the
output~\cite{AhlswedeDueck1989Identificationviachannels,HanVerdu1993Approximationtheoryoutput}.
ID can have applications in various domains that span authentication tasks such as
watermarking~\cite{SteinbergMerhav2001Identificationpresenceside,
  Steinberg2002Watermarkingidentificationprivate_conference},
sensor communication~\cite{GuenlueKliewerSchaeferSidorenko2021DoublyExponentialIdentification_conference},
vehicle-to-X communication~\cite{BocheDeppe2018SecureIdentificationWiretap,BocheArendt2021patentV2Xid,RosenbergerPeregDeppe2022IdentificationoverCompound_conference}, among others.

From a more theoretical point of view, there are as well several interesting connections between ID and common
randomness generation~\cite{AhlswedeCsiszar1998Commonrandomnessinformation},
as well as resolvability~\cite{HanVerdu1993Approximationtheoryoutput,
  Steinberg1998Newconversestheory,
  Hayashi2006Generalnonasymptoticasymptotic,
  Watanabe2022MinimaxConverseIdentification}
and soft-covering~\cite[p.~656]{Ahlswede2018CombinatorialMethodsModels},~
\cite{Wyner1975commoninformationtwo,Cuff15,BocheDeppe2019SecureIdentificationPassive}.

These connections lead to remarkable behavior of ID capacities in
relation
with security constraints.
For example, in the semantic

secrecy
regime, one can achieve the same ID rate,
as if there were no security requirement, provided that
the secrecy-capacity of the channel is positive~\cite{AhlswedeZhang1995Newdirectionstheory}, \cite{wiretapwafa}.
Stealth~\cite{cheJaggiEA2014surveyStealth,HouKramerBloch2017EffectiveSecrecyReliability,songJaggiEA2020stealthMultipathJammed,BlochGuenlueYenerOggierPoorSankarSchaefer2021OverviewInformationTheoretic}
requires that the adversary is prevented
from determining whether Alice and Bob are communicating.
As such, it may seem to be an even stronger requirement than secrecy.
In general, however, neither stealth nor secrecy implies the
other~\cite{HouKramerBloch2017EffectiveSecrecyReliability}.
In contrast to
covert communication\cite{BashGoeckelTowsley2013covertAWGN,BashGoeckelTowsleyGuha2015covertLimits,Bloch2016covertResolvability},
stealthy 
signals are not necessarily limited in power,
but need to simulate a default distribution at the attacker's output,
so that the attacker is unable to distinguish valuable information from noise.
Thus, covert communication belongs to the same family of security problems as
stealthy communication~\cite{HouKramerBloch2017EffectiveSecrecyReliability,LentnerKramer20}.

The covert ID capacity for binary-input discrete memoryless
wire-tap channels was determined by Zhang~and~Tan~\cite{ZhangTan2021CovertIdentificationBinary},
and an achievable rate for ID under semantic effective secrecy
has been established 
\cite{IbrahimFerraraDeppe2021IdentificationEffectiveSecrecy_conference}.

Here, we derive new bounds on the ID capacity of discrete
memoryless wiretap channels~\cite{Wyner1975wiretap,BlochBarros2011phys_sec_book}
under the semantic effective secrecy constraint.
We improve the achievable rate from~\cite{IbrahimFerraraDeppe2021IdentificationEffectiveSecrecy_conference}
by observing that the encoder can always prepend an auxiliary channel
to the actual channel given to him. This may increase the achievable rate, e.g.,
if the channel is not more capable, i.e. $I(X; Y) \le I(X; Z)$ for
the input distribution $P_X$ achieving the capacity without security constraints,
the legitimate channel output $Y$, and attacker's output $Z$, but there exists
an auxiliary random variable $U$ such that $U - X - YZ$ is a Markov chain
and $I(U; Y) \ge I(U; Z)$.
Furthermore, we derive a converse bound,
which is tight if the legitimate channel is more capable
than the eavesdropper's one.
We demonstrate the results for a reversely degraded binary erasure broadcast channel,
where each symbol consists of two bits, and for the first bit, the legitimate channel is
stronger, and for the second bit, the eavesdropper's channel is better.
Based on this example, we discuss the relation of the derived bounds, and the relations
to capacities for other communication problems and constraints.
Finally, we discuss future steps needed to obtain a generally tight converse bound.

The ID capacity of the discrete memoryless wiretap channel
exhibits a similar dichotomy
for semantic effective secrecy
as for only semantic secrecy~\cite{AhlswedeZhang1995Newdirectionstheory},
but with a more stringent positivity condition and constraint.
This is because for secrecy, only
a small part of the ID codeword has to be
secret~\cite{AhlswedeZhang1995Newdirectionstheory},
while for effective secrecy, the whole codeword must be stealthy \cite{IbrahimFerraraDeppe2021IdentificationEffectiveSecrecy_conference}.

This work is organized as follows.
\Cref{sec:prelim} sets the notation, channel model and communication task.
In \cref{sec:results}, we review previous results and present our main theorems.
We present the example and discuss the bounds in \cref{sec:example}.
In \cref{sec:auxResults}, we develop an auxiliary ID
converse, which permits convex constraints on the encoding distributions.
It is then used in \cref{sec:esid.dm.proof} to prove the
upper rate bound for effectively secrete ID.
Finally, \cref{sec:conclusion} summarizes the results and discusses futher steps.

An extended version of this paper can be found on
\uref{https://arxiv.org}{arXiv}.

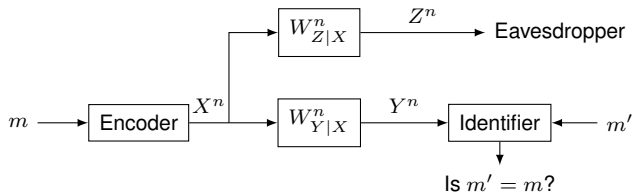
\begin{figure}
	\centering
	\begin{tikzpicture}[scale=1.6]
	
    \path
        node[box] (channel) {$W_{Y|X}^n$}
        ++(0,.75)
        node[box] (wiretap) {$W_{Z|X}^n$}
        ++(2,0)
        node      (eve)     {Eavesdropper}
        (channel)
        ++(-.75,0)
        coordinate (tap)
        (channel)
        ++(-1.5,0)
        node[box] (encoder) {Encoder}
        ++(-1,0)
        node      (message) {$m$}
        (channel)
        ++(1.5,0)
        node[box] (decoder) {Identifier}
        +(0,-.5)
        node      (test)    {Is $m' = m$?}
        ++(1,0)
        node      (rec_message) {$m'$}
    ;

    \draw[pfeil] (message) -> (encoder);
    \draw[pfeil] (encoder) -- node[above] {$X^n$} (tap) -> (channel);
    \draw[pfeil] (tap) |- (wiretap);
    \draw[pfeil] (channel) -> node[above] {$Y^n$} (decoder); 
    \draw[pfeil] (decoder) -> (test);
    \draw[pfeil] (rec_message) -> (decoder);
    \draw[pfeil] (wiretap) -> node[above] {$Z^n$} (eve);
    \end{tikzpicture}
	\caption{\label{fig:cr-scheme}
	  A general identification scheme in the presence of an eavesdropper,
	  where $W_{YZ|X}^n = \prod_{i=1}^n W_{YZ|X}$ is a discrete memoryless wiretap
	  channel. In contrast to conventional message transmission,
	  the receiver does not decode the message $m$ from the channel output
	  $Y^n$, but chooses an $m'$, and performs a statistical hypothesis
	  test to decide whether $m'$ equals $m$ or not.
	  In the effective secrecy setting, the eavesdropper wants to
	  find out whether unexpected communication takes place or not,
	  compared to some expected default behavior, and to
	  identify whether an own message $m''$ equals $m$ or not.
	}
\end{figure}

\section{Preliminaries}
\label{sec:prelim}

The indicator function $\ind{\cdot}$ evaluates to $1$
if its argument is true, and to $0$ if it is false.

The function $\log$ denotes the natural logarithm.

\subsection{Channels}

A \defname{channel} with domain $\cX$ and codomain $\cY$
is a conditional PMF $W : \cX \to \cP(\cY)$.
A \defname[DMC]{discrete memoryless channel $W$} is a sequence
$(W^n)_{n \in \bbN}$, where an input sequence $x^n$
of block length $n \in \bbN$
is mapped to an an output sequence $y^n$
with probability

$
  W^n(y^n|x^n) = \prod_{i=1}^n W(y_i | x_i).
$

A \defname{wiretap channel} is a channel $W_{YZ|X} : \cX \to \cP(\cY \times \cZ)$,
where we assume that for an input $x$,
a legitimate receiver has access to the channel output $Y \sim W_{Y|X=x}$
and a passive adversary has access to to the output $Z \sim W_{Z|X=x}$,
where the output distributions of the marginal channels $W_{Y|X}$ and $W_{Z|X}$
given input $x$ are marginalizations of the joint output distribution $W_{YZ|X=x}$.
A \defname[DMWC]{discrete memoryless wiretap channel $W_{YZ|X}$} is
a sequence $(W_{YZ|X}^n)_{n \in \bbN}$.
A wiretap channel is \defname{(stochastically) degraded} towards $Y$
if there exists a channel $W_{Z|Y}$ such that $W_{Z|X} = W_{Y|X} W_{Z|Y}$.

The \defname[KL divergence]{Kullback-Leibler divergence} between two PMFs
$P,Q \in \cP(\cX)$, where $Q(x) > 0$ if $P(x) > 0$, is defined by

$
  D(P \| Q) = \sum_{x \::\: P(x) > 0} P(x) \log \frac{P(x)}{Q(x)},
$

and two conditional PMFs $W, V : \cX \to \cP(\cY)$,
the conditional KL divergence is defined by
$D(W \| V | P) = \expect_P \brack{ D(W(\cdot|X) \| V(\cdot|X)) }$.
The Shannon entropies and the mutual information are defined as usual.
Note that for $(X,Y) \sim P \times W$ and any $Q \in \cP(\cY)$
such that the following expression is defined, it holds that
\begin{gather}
  D(W \| Q | P) = I(X; Y) + D(PW \| Q).
\end{gather}

\subsection{Identification codes and effective secrecy}

The task of ID with effective secrecy is described as follows:
Consider a wiretap channel $W_{YZ|X}$ where Alice transmits
a signal $X \in \cX$, Bob, the legitimate receiver, receives $Y \in \cY$
and Willie, the adversary, receives $Z \in \cZ$.
Alice encodes a message $m \in [M] = \set{1,\dots,M}$ into $X$ such that
(a) Bob can test reliably whether Alice sent $m'$ or not, for every $m' \in [M]$
\defname{(identification)}
and (b) Willie cannot distinguish whether Alice sent something
sensible or nonsense \defname{(stealth)}, nor
can he identify whether Alice sent any particular $\tilde{m} \in [M]$ \defname{(secrecy)}.
The combination of stealth and secrecy is called \defname{effective secrecy}.
To this end, it is assumed that Willie will always classify the received signal
as suspicious if it differs significantly from a prescribed distribution $Q_Z$.
Let us formally define the involved codes and capacities:

An \defname{$M$-code} for a discrete channel $W_{Y|X} : \cX \to \cP(\cY)$
is a family $\set{ (E_m, \cD_m) }_{m=1}^M$
of encoding distributions $E_m \in \cP(\cX)$
and decision sets $\cD_m \subseteq \cY$.
An $(M, n)$-code for a DMC $W$ is an
$M$-code for the channel $W_n = W^n$,
i.e. $E_m \in \cP(\cX^n)$ and $\cD_m \subseteq \cY^n$.

A \defname[ID code]{$(M | \lambda_1, \lambda_2)$-identification code} is
an $M$-code satisifying the conditions
\begin{align}
  \min_m E_m W_{Y|X} (\cD_m)
    &\ge 1 - \lambda_1,
    \\
  \max_{m, m' \::\: m \neq m'} E_m W_{Y|X} (\cD_{m'})
    &\le \lambda_2.
\end{align}

An
\defname[ESID code]{$(M | \lambda_1, \lambda_2, \delta, Q_Z)$ (semantically) effectively secret ID code}
for a discrete wiretap channel $W_{YZ|X}$
is an $(M | \lambda_1, \lambda_2)$-ID-code
where every encoding distribution $E_m$ simulates $Q_Z$ with precision $\delta > 0$ over
$W_{Z|X}$, i.e.
\begin{gather}
  \max_m D(E_m W_{Z|X} \| Q_Z)
    \le \delta.
\end{gather}

Semantic secrecy means that the constraint must hold for every message,
not only for a particular random message distribution.
Similarly, an $(M, n | \lambda_1, \lambda_2)$ ID code
and an $(M, n | \lambda_1, \lambda_2, \delta, Q_{Z^n})$ ESID code
are defined for DMCs with block length $n$.

The \defname{rate} of an $(M, n)$ ID code is defined as
$R = \frac{1}{n} \log\log M$.
A rate $R$ is $Q_Z$-ESID \defname{achievable} over a wiretap channel $W_{YZ|X}$
if, for all $\lambda_1, \lambda_2, \delta > 0$
and sufficiently large $n$,
there exists an $(2^{2^{nR}}, n | \lambda_1,\lambda_2, \delta, Q_Z)$ ESID code.
The $Q_{Z^n}$-ESID capacity $\capESID(W_{YZ|X}, Q_{Z^n})$
is the supremal rate that is $Q_{Z^n}$-ESID achievable over $W_{YZ|X}$.

\section{Results}
\label{sec:results}

Consider the prior result from \cite{IbrahimFerraraDeppe2021IdentificationEffectiveSecrecy_conference}.
\begin{proposition}[{\cite[Theorem~1]{IbrahimFerraraDeppe2021IdentificationEffectiveSecrecy_conference}}]
  \label{prop:achiev.X}
  For any $Q_Z \in \cP(\cZ)$,
  the $Q_Z^n$-ESID capacity of a discrete memoryless wiretap channel $W_{YZ|X}$ satisfies
  \begin{gather}
    \capESID(W_{YZ|X}, Q_Z^n)
    \ge
      \max_{\substack{
        P_X \in \cP(\cX)
        \\ P_X W_{Z|X} = Q_Z
        \\ I(X; Y) \ge I(X; Z)
      }}
      I(X; Y).
  \end{gather}
\end{proposition}

In \cref{prop:achiev.X}, the constraint $I(X; Y) \ge I(X; Z)$
can be relaxed to $I(U; Y) \ge I(U; Z)$,
for any auxiliary random variable $U$ having finite support and
satisfying the Markov condition $U - X - YZ$,
by applying~\cref{prop:achiev.X} to the virtual channel
$P_{YZ|U} = P_{X|U} W_{YZ|X}$:
\begin{corollary}
  The $Q_Z^n$-ESID capacity of a DMWC $W_{YZ|X}$ satisfies
  \begin{gather}
  \label{eq:achiev.U}
    \capESID(W_{YZ|X}, Q_Z^n)
    \ge
      \max_{\substack{
        P_{UX} \in \cP(\cU \times \cX)
        \\ P_X W_{Z|X} = Q_Z
        \\ I(U; Y) \ge I(U; Z)
      }}
      I(U; Y),
  \end{gather}
  \label{corollary:achiev.U}
  where $\cU$ is any finite set.
\end{corollary}
On the other hand, we prove the following upper bound.
\begin{theorem}
  \label{thm:esid.dm}
  The $Q_Z^n$-ESID-capacity of a DMWC $W_{YZ|X} : \cX \to \cP(\cY \times \cZ)$
  is $0$ if $I(U; Y) < I(U; Z)$ or $P_X W_{Z|X} \ne Q_Z$, for all $P_{UX}$ such that
  $U - X - YZ$ forms a Markov chain. Otherwise, it
  satisfies
  \begin{gather}
  \label{eq:esid.dm}
    \capESID(W_{YZ|X}, Q_Z^n) \le
    \max_{\substack{
      P_{UX} \in \cP(\cU \times \cX)
      \\ P_X W_{Z|X} = Q_Z
      \\ I(U; Y) \ge I(U; Z)
    }}
    I(X; Y),
\end{gather}
  where $\card\cU \le \card\cX + 2$.
  If $W_{Y|X}$ is more capable than $W_{Z|X}$,
  \begin{gather}
    \capESID(W_{YZ|X}, Q_Z^n) =
    \max_{\substack{
      P_X \in \cP(\cX)
      \\ P_X W_{Z|X} = Q_Z
    }}
    I(X; Y).
  \end{gather}
\end{theorem}
The proof follows in Section \ref{sec:esid.dm.proof}.

\begin{remark}
  The lower bound in \cref{corollary:achiev.U}
  and the upper bound in \cref{thm:esid.dm} coincide only
  for channels where the optimal capacity is achieved
  with $U = X$.
  In the following section, we demonstrate the gap
  at the example of a reversely degraded wiretap channel,
  where $U \ne X$ is optimal.
\end{remark}
\begin{remark}
  If the $Q_Z^n$-ESID capacity is zero,
  then effectively secret communication with any positive rate
  impossible.
  Yet, this does not necessarily imply that communication is
  impossible. It can simply mean that the code size grows slower than
  doubly-exponentially in the block length, since we defined the
  rate as $R = \frac{1}{n} \log\log M$.
  
  For example, for covert codes $O(\sqrt n)$ bits can be sent in $n$ channel uses.  However,
  Ahlswede~\cite[Lemmas 89, 90 and Remark 92]{Ahlswede2021IdenticationOtherProbabilistic}
  proved that, if for sufficiently small $\lambda_1,\lambda_2,\delta$
  and sufficiently large $n$,
  the secrecy condition $I(U; Y) \ge I(U; Z)$ is
  violated for all $U$, then secret communication is impossible,
  hence also effectively secret communication.
\end{remark}

\section{Example: Reversely degraded broadcast channels}
\label{sec:example}

\begin{figure}
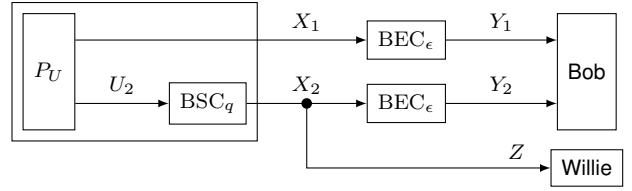

    \centering
    \exampleChannelESID[xscale=1.1,yscale=0.5][1]
    \caption{\label{fig:revDegraded}
      Structure of the product of two reversely degraded broadcast example 
      }
\end{figure}
To demonstrate the relation of the capacity bounds in
\cref{prop:achiev.X,corollary:achiev.U,thm:esid.dm}
and contrast it with message transmission,
we consider two reversely degraded binary broadcast channels
used in parallel,
as shown in \cref{fig:revDegraded}.
In this section, in a binary setting, we let $\log := \log_2$.
Every input symbol consists of two bits,
i.e. $\cX = \set{ 0,1 }^2$.
The bit $X_1$ is sent only to Bob,
over a binary erasure channel, $\BEC_\epsilon$,
described by the transition matrix

$
  \BEC_\epsilon =
  \mleft(\begin{smallmatrix}
    1 - \epsilon  & \epsilon & 0 \\
    0  & \epsilon & 1 - \epsilon
  \end{smallmatrix}\mright),
  $

where $\frac{1}{2} < \epsilon \le 1$
is the erasure probability, the output alphabet
is $\set{0,e,1}$, and $e$ denotes an erasure.
The bit $X_2$ is also sent to Bob over a $\BEC(\epsilon)$ with the same
erasure probability,
while Willie observes it
noiselessly.
We have thereby the assignments
     $X = (X_1, X_2)$,
  $Y = (Y_1, Y_2)$,
  $Z = X_2$,
  $Y_1 \sim \BEC_\epsilon(\cdot|X_1)$, and
  $Y_2 \sim \BEC_\epsilon(\cdot|X_2)$.
The output alphabets of the described channel
are $\cY = \set{0,e,1}^2$ and $\cZ = \cX$.
This channel belongs to the class of reversely degraded broadcast channels \cite[Page 127]{ElGamalKim2011NetworkInformationTheory},
i.e. $P_{Y,Z|X}=P_{Y_1,Y_2,Z|X_1,X_2}=P_{Y_1|X_1}\cdot P_{Z|X_2}\cdot P_{Y_2|Z}$.

Let $P_{X_1} = (p_{X_1}, 1 - p_{X_1})$, i.e. $P_{X_1}(0) = p_{X_1}$,
and similarly $P_{X_2} = (p_{X_2}, 1 - p_{X_2})$.
Since $X_1$ is perfectly secret,
Alice chooses $p_{X_1} = \frac{1}{2}$,
which maximizes the mutual information
$I(X_1; Y_1) = (1-\epsilon) H_2(p_{X_1}) = 1 - \epsilon$,
by \cite[Eq.~(7.15)]{CoverThomas2005ElementsInformationTheory}.
Suppose the default distribution to simulate
for effective secrecy is $Q_Z = (\frac{1}{2}, \frac{1}{2})$.
Then, \cref{eq:achiev.U,eq:esid.dm} require that $p_{X_2} = \frac{1}{2}$.
Thus, the mutual informations of the marginal
channels are
\begin{align}
  I(X; Y)
    & = I(X_1; Y_1) + I(X_2; Y_2)
  \\& = (1-\epsilon) \tup{ H_2(p_{X_1}) + H_2(p_{X_2}) }
  \\& = 2 (1-\epsilon),
  \\
  I(X; Z)
    & = H_2(p_{X_2}) = 1,
\end{align}
where $H_2(p) = -p \log p - (1-p) \log (1-p)$.
Since $\epsilon > \frac{1}{2}$,
$P_{Y|X}$ is less capable than $P_{Z|X}$, i.e.

  $I(X; Y) < I(X; Z),$

and \cref{prop:achiev.X}
guarantees no achievable $Q_Z^n$-ESID rate.
Yet, Alice can achieve
$Q_Z^n$-ESID, by letting $p_{U_2} = \frac{1}{2}$ 

  $P_{U_2}  = (p_{U_2}, 1-p_{U_2})$,
  $X_2  \sim \BSC_q(\cdot|U_2)$,
  and $U    = (X_1, U_2)$,

where
  $\BSC_q =
  \mleft(\begin{smallmatrix}
    1 - q & q \\
    q & 1 - q
  \end{smallmatrix}\mright)
  $
is a binary symmetric channel with crossover probability $q$.
Elementary calculations show that
\begin{align}
  p_{X_2}
    & = p_{U_2} (1 - q) + (1-p_{U_2}) q
    = \frac{1}{2},
  \\
  I(U; Y)
    & = I(X_1; Y_1) + I(U_2; Y_2)
  \\& = (1 - \epsilon) (2 - H_2(q)),
  \\
  I(U; Z)
    & = I(U_2; Z)
   = 1 - H_2(q),
\end{align}
and

$
  I(U; Y) \ge I(U; Z),
$
for all $\epsilon \le 1/(2 - H_2(q))$.

Hence, by \cref{corollary:achiev.U,thm:esid.dm},
\begin{align}
  2(1 - \epsilon)
    & \ge I(X; Y)
  \ge \capESID(W_{YZ|X}, Q_Z^n)
  \\& \ge I(U; Y)
  \\& = (1-\epsilon) (2 - H_2(q)),
\end{align}
where the gap is given by

$
  I(X; Y) - I(U; Y) = (1-\epsilon) H_2(q).
$

For comparison, given $P_{X_2|U_2} = \BSC_q$,
an upper rate bound
for effectively secret message transmission

is
\cite[Theorem~1.1]{HouKramerBloch2017EffectiveSecrecyReliability}
\begin{align}
  R_{\mathsf{EST}}
    & \le I(U; Y) - I(U; Z)
  \\& \le 1 - \epsilon
  \\& \le I(X_1; Y).
\end{align}
On the other hand, the ID capacity with only secrecy,
without stealth, is given by $I(X; Y)$, since there exists $P_U$
such that $I(U; Y) \ge I(U;Z)$
This suggests that for transmission with effective secrecy,
it is optimal for Alice to only encode into the first bit, $X_1$,
while effectively secret ID codes can increase the rate
by exploiting both bits $(X_1, U_2)$.
\Cref{fig:mutInfsRevDegraded}
displays the mutual informations for varying $0 \le p_{U_2} = p_{X_2} \le 1$,
$q = \frac{1}{8}$, and hence
$\epsilon = 1/(2 - H_2(q)) = \frac 3 8 -\frac 7 8 \log_2(\frac 7 8) \approx 0.6866$. Thus, for $p_{U_2} = \frac{1}{2}$,
we have $I(U; Y) = I(U; Z)$.

\begin{figure}
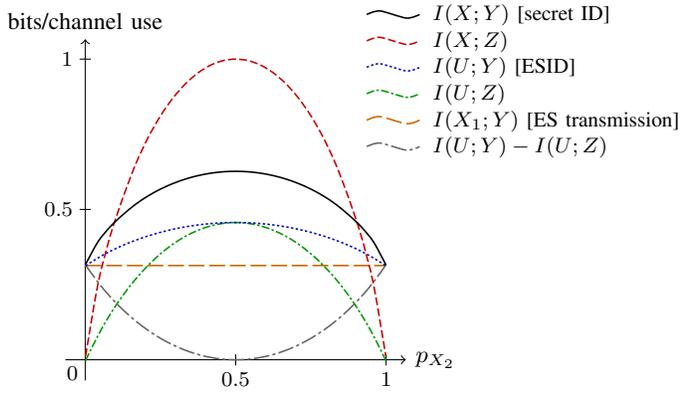

    \centering
    \tikzMutInfsRevDegraded{1/8}{0.6866}
    \caption{\label{fig:mutInfsRevDegraded}
      Mutual informations for the reversely degraded wiretap channel
      in \cref{fig:revDegraded}, and communication tasks
      where the mutual information is achievable,
      where $q = \frac{1}{8}$, $\epsilon = 0.6866 \approx 1/(2-H_2(q))$,
      and $P_{X_1} = Q_Z = (\frac 1 2,\frac 1 2)$.
      For secret ID and effectively secret message transmission,
      the given bounds are tight, if $p_{U_2} = \BSC(q)$, for any $q$.
    }
\end{figure}

\section{Auxiliary identification converse}
\label{sec:auxResults}

Consider the hypothesis testing divergence

\begin{gather}
  D_\alpha(P \| Q)
  := \sup\set{ \gamma : P\tup{ \log \frac{P(Y)}{Q(Y)} \le \gamma } \le \alpha },
\end{gather}
where $Y \sim P$,
and let $(P \otimes Q) (x, y) := P(x) Q(y)$.

\begin{lemma}
  \label{lemma:idConverse.oneShot}
  Let $\lambda_1, \lambda_2, \eta > 0$ and $\alpha := \lambda_1 + \lambda_2 + 2\eta < 1$.
  For every $(M | \lambda_1, \lambda_2)$ ID code $\set{(E_m, \cD_m)}_{m=1}^M$
  for a channel $W_{Y|X} : \cX \to \cP(\cY)$,
  its size is bounded by
  \begin{align}
    \log\log M
    &
    \label{eq:idConverse.oneShot.divergence.PX}
    \le \max_{P_X \in \cE} \min_{Q \in \cP(\cY)}
      D_\alpha (P_{XY} \| P_X \otimes Q)
      + \epsilon
    \\&
    \label{eq:idConverse.oneShot.mutinf}
    \le \max_{P_X \in \cE}
      \frac{1}{1-\alpha}
      I(X; Y)
      + \epsilon,
  \end{align}
  where $P_{XY} = P_X \times W_{Y|X}$, $\epsilon = \log\log \card\cX + 3 \log(1/\eta) + 2$,
  and $\cE = \set{ P_V E_{m=V} : P_V \in \cP([M])}$
  is the convex hull of all encoding distributions.
\end{lemma}

\begin{proof}
  By 
  \cite[Corollary~2 and~Lemma~1]{Watanabe2022MinimaxConverseIdentification},
  we have that
  \begin{gather}
    \label{eq:idConverse.oneShot.divergence.PX.orig}
    \log\log M
    \le \max_{P_X \in \cP(\cX)} \min_{Q \in \cP(\cY)}
      D_\alpha (P_{XY} \| P_X \otimes Q)
      + \epsilon.
  \end{gather}
  The maximization over $\cP(\cX)$ is introduced in
  the proof of \cite[Theorem~1]{Watanabe2022MinimaxConverseIdentification}
  to establish the upper bound
  \begin{align}
    \label{eq:watanabe.bound}
    \frac{1}{2} &\brack{ (E_m \times W_{Y|X}) (\cS) + (E_{m'} \times W_{Y|X}) (\cS) }
    \nonumber
    \\
    &\le
    \sup_{P_X \in \cP(\cX)} P_X \times W_{Y|X} (\cS),
  \end{align}
  for some set $\cS \subseteq \cX \times \cY$,
  which ultimately leads to the maximization in
  \cref{eq:idConverse.oneShot.divergence.PX.orig}.
  Clearly, a maximization over all $E_m$, $m \in [M]$
  would suffice in \cref{eq:watanabe.bound}.
  
  To
  establish the minimax equality in \cite[Corollary~2]{Watanabe2022MinimaxConverseIdentification},
  Watanabe used the fact that the supremum is taken over
  a compact and convex set. Since no 
  other
  properties
  of $\cP(\cX)$ are used in~\cite{Watanabe2022MinimaxConverseIdentification},
  it suffices to maximize over the convex hull $\cE$
  of all encoding distributions, 
  
  as in~\cref{eq:idConverse.oneShot.divergence.PX}.
  By Markov's inequality,
  \begin{align}
    D_\alpha(P \| Q)
    &
      = \sup\set{ \gamma : P\tup{ \log \frac{P(Y)}{Q(Y)} > \gamma } \ge 1 - \alpha }
    \\&
      \le \inf\set{ \gamma : P\tup{ \log \frac{P(Y)}{Q(Y)} \ge \gamma } \le 1 - \alpha }
    \\&
      \le \frac{1}{1-\alpha} \expect\intv{ \log\frac{P(Y)}{Q(Y)} }
    \\&
      \le \frac{1}{1-\alpha} D(P\|Q),
  \end{align}
  where the first inequality holds since $p(\gamma) = P\tup{\log\frac{P(Y)}{Q(Y)} \ge \gamma}$ is a decreasing function.
  Thus, for every $P_X \in \cP(\cX)$,
  \begin{align}
    \min_{Q \in \cP(\cY)}
    \nonumber
      & D_\alpha(P_X \times W_{Y|X} \| P_X \otimes Q)
    \\& \le \frac{1}{1-\alpha} D(P_X \times W_{Y|X} \| P_X \otimes P_X W_{Y|X})
    \\& =   \frac{1}{1-\alpha} I(X; Y),
  \end{align}
  and \cref{lemma:idConverse.oneShot} follows.
\end{proof}

\section{Proof of \cref{thm:esid.dm}}
\label{sec:esid.dm.proof}

Consider any
$(M, n | \lambda_1,\lambda_2,\delta,Q_Z)$ ESID code
$\set{ \tup{E_m, \cD_m} }_{m=1}^M$
for a DMWC
$W_{YZ|X} : \cX \to \cP(\cY \times \cZ)$,
where $\lambda_1, \lambda_2, \eta > 0$
satisfy $\alpha := \lambda_1 + \lambda_2 + 2\eta < 1$.
By \cref{lemma:idConverse.oneShot}, 
the rate is upper-bounded by
\begin{gather}
  \label{eq:esid.converse.oneShot.rate}
  R = \frac{1}{n} \log\log M
  \le \max_{P_{X^n} \in \cE}
    \frac{1}{n(1-\alpha)} I(X^n; Y^n) + \frac{\epsilon}{n},
\end{gather}
where $\epsilon = \log\log \card{\cX^n} + 3\log(1/\eta) + 2$,
and $\cE$ is the convex hull of all encoding distributions.
By the chain rule,
\begin{align}
  \frac{1}{n} I(X^n; Y^n)
    & = \frac{1}{n}\sum_{i=1}^n I(X^n; Y_i | Y^{i-1})
  \\& =
  I(X^n; Y_T | T
  , Y^{T-1})
  \\& \le I(X_T; Y_T),
  \label{eq:singleLetter.mutInfXY}
\end{align}
where $T \sim P_T(i) = \frac 1 n \ind{1 \le i \le n}$,
and \cref{eq:singleLetter.mutInfXY} follows from the
concavity of the mutual information in the input argument, and the
Markov condition $Y^{i-1} - X^{i-1}X_{i+1}^n - X_i - Y_i$

for every $i \in [n]$.

In the following, we single-letterize the constraints
on $\cE$.
Similarly to \cite[Eq. (1.49)]{HouKramerBloch2017EffectiveSecrecyReliability},
for any $P_{Z^n} \in \cP(\cZ^n)$,
\begin{align}
  \delta & = D(P_{Z^n} \| Q_Z^{n})
  \\& = \sum_{z^n} P_{Z^n}(z^n) \sum_{i=1}^n \log \frac{1}{Q_Z(z_i)} - H(Z^n)
    \label{eq:esid.converse.dm.expandStealth}
  \\& \ge \sum_{i=1}^n \sum_z P_{Z_i}(z) \log \frac{1}{Q_Z(z)} - \sum_{i=1}^n H(Z_i)
  \\& = n \sum_{i=1}^n P_T(i) D(P_{Z_i} \| Q_Z)
  \\& \ge n D(P_{Z_T} \| Q_Z).
    \label{eq:esid.converse.dm.stealthSingleLetter}
\end{align}
By \cite[Lemma~90]{Ahlswede2021IdenticationOtherProbabilistic},
for sufficiently small $\lambda_1, \lambda_2, \delta > 0$,
\begin{align}
  \smashoperator{\max_{P_A \in \cP([M])}} ~ I(A; Y^n)
    & \ge \delta
  \\& \ge \max_{m \in [M]} D(E_m W_{Z|X}^n \| Q_Z^n)
  \\& \ge \max_{P_A \in \cP([M])} D(E_{m=A} W_{Z|X}^n \| Q_Z^n | P_A)
  \\& \ge \max_{P_A \in \cP([M])} I(A; Z^n).
\end{align}
Thus, there exists $P_A \in \cP([M])$ such that
\begin{align}
  0 & \le \frac{1}{n} [ I(A; Y^n) - I(A; Z^n) ]
    \label{eq:esid.converse.dm.substractMutInf}
  \\& = I(V A; Y_T|V) - I(V A; Z_T|V)
    \label{eq:esid.converse.dm.singleLetter}
  \\& \le \max_v \max_{P_{BX|V=v}} [ I(B; Y_T|V=v) - I(B; Z_T|V=v) ]
    \label{eq:esid.converse.dm.diff.maxV}
  \\& \le \max_{P_{BX}}  [ I(B; Y) - I(B; Z) ],
    \label{eq:esid.converse.dm.rateBound}
\end{align}
where
$B = (V,A)$,
$V = (T, Y_1,\dots,Y_{T-1}, Z_{T+1},\dots,Z_n)$,
the maximizations are with respect to 
$D(P_B P_{X_T|B} W_{Z|X}\| Q_Z) \le \frac{\delta}{n}$,
and \cref{eq:esid.converse.dm.singleLetter} follows
from~\cite[Lemma 17.12]{CsiszarKoerner2011InformationTheoryCoding}.
By~\cite[Lemmas~15.4 and~15.5]{CsiszarKoerner2011InformationTheoryCoding},
we can replace $P_B$ by $P_U \in \cP(\cU)$, $\card\cU \le \card\cX + 2$,
such that
$I(U; Y) = I(B; Y)$,
$I(U; Z) = I(B; Z)$,
and $P_U P_{X|B} = P_B P_{X|B}$.

Since the mutual information is continuous and
the set $\set{ P_{UX} : D(P_X W_{Z|X} \| Q_Z) \le \delta }$
is compact, for $\eta = e^{-\sqrt{n}}$, we have that
\begin{align}
  &\capESID(W_{YZ|X}, Q_Z^n)
  \\&
    \le
    \inf_{\lambda_1,\lambda_2,\delta > 0}
    \lim_{n \to \infty}
    \brack[\Bigg]{
      o(1)
      +
      \frac{1}{1 - \lambda_1 - \lambda_2 - o(1)}
  \nonumber\\&\hspace{7em}
      \max_{\substack{
        P_{UX} \in \cP(\cU \times \cX)
        \\ D(P_X W_{Z|X} \| Q_Z) \le \delta/n
        \\ I(U; Y) \ge I(U; Z)
      }}
      I(X; Y)
    }
  \\&
    = \max_{\substack{
        P_{UX} \in \cP(\cU \times \cX)
        \\ P_X W_{Z|X} = Q_Z
        \\ I(U; Y) \ge I(U; Z)
      }}
      I(X; Y).
\end{align}
For more capable channels, $I(X; Y) \ge I(X; Z)$, for all $P_X \in \cP(\cX)$
and hence, the upper bound is achievable,
by \cref{prop:achiev.X}.
This completes the proof of \cref{thm:esid.dm}.
\qed

\section{Conclusion}
\label{sec:conclusion}

In \cref{corollary:achiev.U}, we improved the lower bound on the
$Q_Z^n$-ESID capacity in the case $I(X; Y) < I(X; Z)$,
where $Q_Z^n$ is a product distribution and $P_X$
satisfies $P_X W_{Z|X} = Q_Z$.
In \cref{thm:esid.dm}, we complement this result by an upper bound
that is tight if $I(X; Y) \ge I(X; Z)$.
The example in \cref{sec:example} illustrates that
in case $I(X; Y) < I(X; Z)$,
the achievability gap between
\cref{corollary:achiev.U} and \cref{thm:esid.dm}
can be substantial, as is the rate advantage of ID
compared to message transmission.
It seems likely that the lower bound in \cref{corollary:achiev.U}
is tight, by results from resolvability theory \cite{HouKramer2013divergence},
since the whole codeword is subject to the stealth constraint.
The difficult part in finding a more stringent converse bound
seems to be the introduction of an auxiliary channel, as
demonstrated in \cref{sec:example}, where the number of possible input
sequences is suitably bounded (see \cref{lemma:idConverse.oneShot}
and the discussion of the gap in Ahlswede's broadcast converse in \cite[Section~2.4]{Bracher2016IdentificationZeroError_phd}).
This is a non-trivial task for ID,
and the authors are not aware of any ID capacity result
involving auxiliary variables in the rate bound.
Usually, in converse proofs \cite{ElGamalKim2011NetworkInformationTheory},
the message is obtained as an auxiliary variable from Fano's inequality,
and then is single-letterized.
However, the mutual information between the message and the channel output
cannot be an upper bound to the ID capacity,
since  ID codes transmit mainly randomness, and the capacity of memoryless
channels is achieved with codes, where only few ($\sqrt{n}$) codeword symbols depend
on the message at all \cite{AhlswedeDueck1989Identificationpresencefeedback}.
To close the achievability gap, new methods need to be developed
to introduce auxiliary variables in ID converse bounds.
This would also be a crucial step in the development of further multi-user converses
for ID and many other communication tasks\cite{Ahlswede2008Generaltheoryinformation,Ahlswede2021IdenticationOtherProbabilistic},
e.g. for the broadcast channel \cite{Ahlswede2008Generaltheoryinformation,Bracher2016IdentificationZeroError_phd}.
To this end, observe that any suitably bounded auxiliary channel forms a polytope
with extremal points $P_{X|U=u}$, $u \in \cU$,
where $\cU$ is the auxiliary alphabet.
This polytope must include the set of stealthy encoding distributions.
If such a pre-channel exists,
the usual ID converse can be applied to it, to obtain an upper bound
that matches the bound in \cref{corollary:achiev.U}.

\ifblind\else
\section*{Acknowledgement}

The authors thank Constantin Runge (Technical University of Munich)
for helpful discussions.

J. Rosenberger, A. Ibrahim and C. Deppe
acknowledge the financial support by the Federal Ministry of
Education and Research of Germany in the program of “Souverän. Digital.
Vernetzt.” Joint project 6G-life, project identification number: 16KISK002. C. Deppe and R. Ferrara were further supported in part by the
BMBF within the grant 16KIS1005. C. Deppe was also supported by the DFG
within the project DE1915/2-1.
B. Bash acknowledges the support from the US National Science Foundation under Grant CCF-2006679.
U. Pereg acknowledges the financial support of the Israel VATAT Junior Faculty
Program in Quantum Science and Technology and the German-Israeli Project
Cooperation (DIP).

\fi

\printbibliography

\end{document}